\documentclass[11pt,a4paper]{article}

\usepackage{a4wide}
\usepackage{amsmath,amssymb,amsthm}
\theoremstyle{plain}
\newtheorem{theorem}{Theorem}[section]
\newtheorem*{theorem*}{Theorem}
\newtheorem{corollary}[theorem]{Corollary}
\newtheorem{conjecture}[theorem]{Conjecture}
\newtheorem{lemma}[theorem]{Lemma}

\theoremstyle{definition}
\newtheorem{definition}[theorem]{Definition}

\newtheorem*{note*}{Note}
\newtheorem*{notes*}{Notes}
\newtheorem*{null*}{}

\newtheorem*{remark*}{Remark}

\usepackage{caption}
\usepackage[shortlabels]{enumitem}
\usepackage{eukdate}
\usepackage[T1]{fontenc}
\usepackage[pdftex]{graphics}
\usepackage[pdftex,pagebackref=true,colorlinks,linkcolor=blue,citecolor=blue]{hyperref}
\renewcommand*{\backref}[1]{}
\renewcommand*{\backrefalt}[4]{[{%
    \ifcase #1 Not cited.%
    \or Cited on page~#2.%
    \else Cited on pages #2.%
    \fi%
  }]}
\usepackage{ifthen}
\usepackage{lmodern}
\usepackage{mathpartir}
\usepackage[sc]{mathpazo}
\usepackage{mathtools}
\usepackage{mdframed}
\usepackage{natbib}
\usepackage{proof}
\usepackage{relsize}
\usepackage{stmaryrd}
\usepackage{url}
\usepackage{wrapfig}

\newcommand{\cubicalrefslist}{bezem2014model, angiulicartesian, awodey2016cubical, pitts2015nominal, birkedal2016guarded}
\newcommand{\cubicalrefs}{\cite{CoquandT:cubttc, \cubicalrefslist}}
\newcommand{\cubicalrefsnocubttc}{\cite{\cubicalrefslist}}

\newcommand{\Bool}{\text{\normalfont\texttt{\{0{,}1\}}}}

\newcommand{\coerce}{\kw{coerce}}
\newcommand{\Cof}{\kw{Cof}}

\newcommand{\comp}{\circ}

\DeclareMathOperator{\Comp}{\kw{Comp}}
\newcommand{\conc}{\mathbin{@}}
\DeclareMathOperator{\Contr}{\kw{Contr}}

\newcommand{\defeq}{\triangleq}

\newcommand{\disj}{\vee}

\DeclareMathOperator{\elim}{\kw{elim}}

\DeclareMathOperator{\Ext}{\kw{Ext}}
\newcommand{\exto}{\nearrow}

\newcommand{\False}{\bot}

\DeclareMathOperator{\Fib}{\kw{Fib}}
\DeclareMathOperator{\fil}{\kw{fill}}

\DeclareMathOperator{\fst}{\kw{fst}}

\newcommand{\fun}{\mathbin{\shortrightarrow}}

\newcommand{\I}{\kw{I}}

\newcommand{\id}{\mathit{id}}

\newcommand{\IdU}[2]{#1 \path_{\Univ} #2}
\newcommand{\IdCong}[2]{#1 \path_{\cong} #2}

\newcommand{\inv}[1]{\overline{#1}}
\DeclareMathOperator{\isFib}{\kw{isFib}}

\newcommand{\join}{\cup}
\newcommand{\kw}[1]{\mathtt{#1}}

\newcommand{\mono}{\rightarrowtail}

\renewcommand{\path}{\sim}

\newcommand{\pos}[1]{\kw{ax_{#1}}}

\newcommand{\Prop}{\Omega}

\DeclareMathOperator{\refl}{\kw{refl}}
\newcommand{\resby}{|}

\DeclareMathOperator{\snd}{\kw{snd}}
\newcommand{\src}{\kw{0}}

\newcommand{\tgt}{\kw{1}}

\newcommand{\Univ}{\mathcal{U}}

\newcommand{\y}{\mathrm{y}}

\newcommand{\by}[1]{\text{by #1}}
\newcommand{\Equiv}[2]{#1 \simeq #2}

\newcommand{\funext}{\mathit{funext}}
\newcommand{\happly}{\mathit{happly}}
\newcommand{\idtoeqv}{\mathit{idtoeqv}}
\newcommand{\isEquiv}[1]{\mathit{isEquiv}(#1)}

\newcommand{\ua}{\mathit{ua}}
\newcommand{\U}{\mathit{U}}
\newcommand{\UA}{\mathit{UA}}

\newcommand{\ct}{%
  \mathchoice{\mathbin{\raisebox{0.5ex}{$\displaystyle\centerdot$}}}%
             {\mathbin{\raisebox{0.5ex}{$\centerdot$}}}%
             {\mathbin{\raisebox{0.25ex}{$\scriptstyle\,\centerdot\,$}}}%
             {\mathbin{\raisebox{0.1ex}{$\scriptscriptstyle\,\centerdot\,$}}}
}
\newcommand{\opp}[1]{\mathord{{#1}^{-1}}}
\let\rev\opp

\title{Decomposing the Univalence Axiom}

\author{Ian Orton\footnote{Supported by a UK EPSRC
  PhD studentship, funded by grants EP/L504920/1, EP/M506485/1.}\; and Andrew M. Pitts\\
  University of Cambridge Computer Laboratory\\
  Cambridge CB3 0FD, UK\\
\texttt{$\{$Ian.Orton,Andrew.Pitts$\}$@cl.cam.ac.uk}}

\date{}

\begin{document}

\maketitle

\begin{abstract}
  This paper investigates Voevodsky's univalence axiom in intensional Martin-L\"of type theory. In particular, it looks at how univalence can be derived from simpler axioms. We first present some existing work, collected together from various published and unpublished sources; we then present a new decomposition of the univalence axiom into simpler axioms.  We argue that these axioms are easier to verify in certain potential models of univalent type theory, particularly those models based on cubical sets. Finally we show how this decomposition is relevant to an open problem in type theory.
\end{abstract}


\section{Introduction}
\label{sec:int}


Extensionality is a principle whereby two mathematical objects are deemed to be equal if they have the same observable properties. Often, formal systems for mathematics will include axioms designed to capture this principle. In the context of set theory we have the axiom of extensionality, which tells us that two sets are equal if they contain the same elements. In the context of (univalent) type theory we have Voevodsky's univalence axiom \cite[Section 2.10]{HoTT}, which tells us, roughly speaking, that two types are equal if they are isomorphic.

In axiomatic set theory the axiom of extensionality is easily formalised as a simple implication $\forall A \forall B ((\forall X (X \in A \iff X \in B)) \implies A = B)$. The converse implication follows from the properties of equality, and combining these two implications we deduce that equality of two sets is logically equivalent to them having the same elements. At this point we are done, we now have an extensionality principle for sets and nothing further needs to be assumed.

The situation is more complicated in a proof relevant setting such as intensional type theory. As with sets we can formalise the statement of interest in the language of type theory as $(A\, B : \U) \to A \simeq B \to A = B$, where $A \simeq B$ is the type of \emph{equivalences} between $A$ and $B$. We then postulate the existence of a term witnessing this statement. As before, the converse implication follows from the properties of the identity type, and hence equality of two types is logically equivalent to them being isomorphic. However, the proof relevant nature of type theory means that what we have described so far will be insufficient. We may want to know how equalities derived using the postulated term compute when passed to the eliminator for identity types. For example, if we convert them back into equivalences do we always get the same equivalence that we started with?

In univalent type theory (UTT), also known as homotopy type theory (HoTT), these problems are resolved by taking a different approach to the statement of the univalence axiom. As mentioned before the converse implication, $(A\, B : \U) \to A = B \to A \simeq B$, follows from the properties of the identity type. The approach taken in UTT is to state that for any types $A$ and $B$ this map is itself an equivalence between the types $A = B$ and $A \simeq B$. From this fact we can deduce the existence of a map in the other direction (the original implication of interest), as well as some information about how that map computes.

Merely stating that a certain canonical map is an equivalence is a very concise way to express the univalence axiom. From a mathematical point of view it is appealingly simple and yet powerful. In particular, this statement has the nice property that it is a \emph{mere proposition} \cite[Definition 3.3.1]{HoTT} and so there is no ambiguity about the term witnessing the axiom.

However, there are some disadvantages to this way of stating the univalence axiom. For example, verifying the univalence axiom in a model of type theory can be a difficult task. Fully expanded, this seemingly simple statement becomes very large, with many complex subterms. Verifying univalence directly, by computing the interpretation of the statement in the model and explicitly constructing the interpretation of its proof term, may be unfeasible. Instead, one would need to build up several intermediate results about contractibility, equivalences, and possibly new constructions such as \emph{Glueing} \cite[Section 6]{CoquandT:cubttc}, through a mixture of internal, syntactic and semantic arguments.

The contribution in this paper is a reduction of the usual statement of univalence to a collection of simpler axioms which are more easily verified in certain models of dependent type theory, particularly those based on cubical sets \cubicalrefs. Importantly, we do not propose these axioms as an alternative statement for the univalence axiom when doing mathematics in univalent type theory. These axioms are designed with the previous goal in mind and are not intended to be mathematically elegant or user-friendly.

In the rest of this paper we begin with some preliminary definitions and notational conventions (Section \ref{sec:prelim}). We then briefly discuss the univalence axiom (Section \ref{sec:univ}). These sections cover existing work. We then introduce our alternative set of axioms (Section \ref{sec:axi}), and examine their application to models of type theory (Section \ref{sec:cubical}). Finally, we propose another application of these axioms  to an open problem in UTT (Section \ref{sec:open}).

\paragraph*{\textbf{Agda formalisation}} 
This work presented in this paper is supported by two separate developments in the Agda proof assistant \cite{Agda}. The first covers the material in sections \ref{sec:prelim}-\ref{sec:axi}, where Agda is useful for precisely tracking universe levels in many of the theorems. The second covers the material in Section \ref{sec:cubical}, and builds on the development accompanying \cite{PittsAM:aximct}. In this development we use Agda to verify our constructions in the internal type theory of the cubical sets topos. The source for both can be found at \url{https://github.com/IanOrton/decomposing-univalence}.

\section{Preliminaries}\label{sec:prelim}
In most of this paper we work in intensional Martin-L\"of type theory with dependent sums and products, intentional identity types, and a cumulative hierarchy of universes $\U_0 : \U_1 : \U_2 : ...$.

We use the symbol $=$ for the identity type, $\equiv$ for definitional equality and $\defeq$ when giving definitions. Given $p : x = y$ and $q : y = z$, we write $p \ct q : x = z$ for the composition of identity proofs, and $\rev{p} : y = x$ for the inverse proof.

We also assume the principle of function extensionality, which states that two functions $f, g : \prod_{x : A} B(x)$ are equal whenever they are pointwise equal: $f \sim g \defeq \prod_{x : A} f(x) = g(x)$. That is, that there exists a term:
\begin{align*}
	\funext_{i,j}:&\prod_{A : \U_i}\;\prod_{B : A \to \U_j }\;\prod_{f, g : \Pi_{x : A}B(x)} f \sim g \to f = g
\end{align*}
for all universe levels $i, j$.

Note that in the Agda development mentioned previously we do not assume function extensionality in general, but rather we make it an explicit assumption to each theorem. This means that we can see exactly where function extensionality is used, and at which universe levels it needs to hold.

We now recall some standard definitions and results in UTT/HoTT. 

\begin{definition}[Contractibility]\label{def:contr}
	A type $A$ is said to be \emph{contractible} if the type
	\[isContr(A) \defeq \sum_{a_0 : A}\prod_{a : A} (a_0 = a)\]
	 is inhabited. Contractibility expresses the fact that a type has a unique inhabitant.
\end{definition} 

\begin{definition}[Singletons]
	Given a type $A$ and element $a : A$, we can define
	\[sing(a) \defeq \sum_{x : A} (a = x)\]
	 to be the type of elements of $A$ which are equal to $a$. It is easily shown by path induction that the type $sing(a)$ is always contractible.
\end{definition}

\begin{definition}[Equivalences]
	An \emph{equivalence} from $A \simeq B$ is a pair $(f,e)$ where $f : A \to B$ and $e$ is a proof that for every $b : B$ the fiber of $f$ at $b$ is contractible. To be precise:
	\[ A \simeq B \defeq \sum_{f : A \to B} \isEquiv{f} \]
	where
	\[ fib_f(b) \defeq \sum_{a : A} (f\; a = b)
	  \qquad\text{and}\qquad \isEquiv{f} \defeq \prod_{b : B} \mathit{isContr}(fib_f(b))\]
	for $A : \U_i$, $B : \U_j$ for any $i,j$. 
\end{definition}

A simple example of an equivalence is the identity function $\id_A : A \to A$ for any type $A$. To demonstrate that $\id_A$ is an equivalence we must show that $\prod_{a : A} \mathit{isContr}(\sum_{x : A} (a = x))$, but this is equivalent to the statement that $sing(a)$ is contractible for all $a : A$.

\section{Voevodsky's Univalence Axiom}
\label{sec:univ}

In this section we introduce Voevodsky's univalence axiom. We then present an existing result which decomposes the univalence axiom into a ``naive'' form and a computation rule. In Section \ref{sec:axi}, we will then decompose these two axioms further into five even simpler axioms.

\begin{definition}[Coerce and idtoeqv]\label{def:coerce}
	For all $i$, and types $A,B:U_i$, there is a canonical map $\idtoeqv : (A = B) \to (A \simeq B)$ which is defined by path induction on the proof $A = B$:
	\[ \idtoeqv(\refl) \defeq \id_A \]
	where $\id_A : \Equiv{A}{A}$ is the identity map regarded as an equivalence. We can also define a map $\coerce : (A = B) \to A \to B$ either by path induction, or as:
	\[ \coerce(p,a) \defeq \fst(\idtoeqv(p))(a) \]
	where $\fst$ is the first projection.
\end{definition}

\begin{definition}[Voevodsky's univalence axiom]\label{def:proper-univ}
	The univalence axiom for a universe $\U_i$ asserts that for all $A , B : \U_i$ the map $\idtoeqv : (A = B) \to (A \simeq B)$ is an equivalence.
\end{definition}

 In light of the following definition we will often refer to the univalence axiom as the \emph{proper} univalence axiom.

\begin{definition}[The naive univalence axiom]\label{def:naive-univ}
	The naive univalence axiom for a universe $\U_i$ gives, for all $A , B : \U_i$, a map from equivalences to equalities. In other words, it asserts the existence of an inhabitant of the type:
	\[ \UA_i \defeq \prod_{A, B : \U_i}\Equiv{A}{B} \to A = B \]
	When using a term $\ua : \UA_i$ we will often omit the first two arguments ($A$ and $B$). Proofs of naive univalence may also come with an associated computation rule. That is, an inhabitant of the type $\UA\beta_i(\ua)$, where:
	\[ \UA\beta_i(\ua) \defeq \prod_{A, B : \U_i}\;\prod_{f : A \to B}\;\prod_{e : \isEquiv{f}} \coerce\; (\ua(f,e)) = f \]
\end{definition}

Next, we give a result which is known in the UTT/HoTT community and has been discussed on the HoTT mailing list. However, the authors are not aware of any existing presentation of a proof in the literature. This result decomposes the proper univalence axiom into the naive version and a computation rule. First we give a lemma which generalises the core construction of this result.

\begin{lemma}\label{lemma:id-retract}
Given $X : \U_i$, $Y : X \to X \to \U_j$ and a map $f : \prod_{x, x' : X}\; x = x' \to Y(x, x')$
then $f\, x\, x'$ is an equivalence for all $x,x' : X$ iff there exists a map
  \[ g : \prod_{x, x' : X}\; Y(x,x') \to x = x' \]
such that for all $x, x' : X$ and $y : Y(x, x')$ we have $f(g(y)) = y$ (we leave the first two arguments to $f$ and $g$ implicit).
\end{lemma}
\begin{proof}
For the backwards direction, assume that we are given $g$ as above. To show that $f$ is an equivalence it suffices to show that $f$ is a bi-invertible map \cite[Section 4.3]{HoTT}. To do this we must exhibit both a right and left inverse.

For the left inverse we take $g'(y) \defeq g(y)\, \ct\, \opp{g(f(\refl))}$. To see that this is indeed a left inverse to $f$ consider an arbitrary $p : x = x'$, we aim to show that $g'(f(p)) = p$. By path induction we may assume that $x \equiv x'$ and $p \equiv \refl$ and therefore we are required to show $g'(f(\refl)) = \refl$. However, since $g'(f(\refl)) \equiv g(f(\refl))\, \ct\, \opp{g(f(\refl))}$ this goal simplifies to $g(f(\refl))\, \ct\, \opp{g(f(\refl))} = \refl$ which follows immediately from the groupoid laws for identity types.

For the right inverse we take $g$ unchanged and observe that we know $f(g(y)) = y$ for all $y : Y(x,x')$ by assumption. Therefore the map $f$ is an equivalence.

For the forwards direction, given a proof $e : \isEquiv{f}$ and $y : Y(x,x')$ we have $\fst(e(y)) : \sum_{p : x = x'} f(p) = y$. We can then define $g(y)$ to be the first component of this and the second component tells us that $f(g(y)) = y$ as required.
\end{proof}

\begin{theorem}\label{naiveuniv}
Naive univalence, along with a computation rule, is logically
equivalent to the proper univalence axiom. That is, there are terms
\[
	\ua : \UA_i, \qquad
	\ua\beta : \UA\beta_i(\ua)
\]
iff for all types $A, B : \U_i$, the map
$\idtoeqv : (A = B) \to (\Equiv{A}{B})$ is an equivalence.
\end{theorem}
\begin{proof}
By $\ua\beta$ we know that $\fst(\idtoeqv(\ua(f,e))) = f$ for all $(f,e) : \Equiv{A}{B}$. Now, since $\isEquiv{f}$ is a mere proposition for each $f$, we can deduce that $\idtoeqv(\ua(f,e)) = (f,e)$ by \cite[Lemma 3.5.1]{HoTT}. Therefore we simply take $X \equiv U_i$, $Y(A,B) \equiv \Equiv{A}{B}$, $f \equiv \idtoeqv$ and $g \equiv \ua$ in Lemma \ref{lemma:id-retract} to deduce the desired result.
\end{proof}


\section{A new set of axioms}
\label{sec:axi}

In this section we further decompose the univalence axiom into even simpler axioms. We show that it is equivalent to axioms (1) to (5) given in Table \ref{table:axioms}. Note that these axioms apply to a specific universe $\U_i$.

\begin{table}[!htbp]
\renewcommand{\arraystretch}{1.1}
\begin{tabular}{|cccccrcl|}
    \hline
     & Axiom && Premise(s) && \multicolumn{3}{c|}{Equality}\\
    \hline
    (1) & $\mathit{unit}$ &:&&&$A$ &=& $\sum_{a : A} 1$\\
    (2) & $\mathit{flip}$ &:&&&$\sum_{a:A}\sum_{b:B} C\; a\; b$ &=& $\sum_{b:B}\sum_{a:A} C\; a\; b$\\
    (3) & $\mathit{contract}$ &:&$isContr\; A$ &$\to$&$A$ &=& $1$\\ 
    (4) &$\mathit{unit\beta}$&:&&&$\mathit{coerce}\; \mathit{unit}\; a$ &=& $(a , *)$\\
    (5) &$\mathit{flip\beta}$&:&&&$\mathit{coerce}\; \mathit{flip}\; (a,b,c)$ &=& $(b , a , c)$\\
  \hline
\end{tabular}\\
\caption{($A, B : \U_i$, $C : A \rightarrow B \rightarrow \U_i$, $a:A$, $b:B$ and $c:C\, a\, b$, for some universe $\U_i$)}
\label{table:axioms}
\end{table}

We begin by decomposing naive univalence, $\UA_i$, into axioms (1)-(3). These axioms also follow from $\UA_i$. Recall that we are taking function extensionality as an ambient assumption.

\begin{theorem}\label{ua}
	Axioms (1)-(3) for a universe $\U_i$ are together logically equivalent to $\UA_i$.
\end{theorem}
\begin{proof}
	We begin by showing the forwards direction. Assume that we are given axioms (1) to (3). We now aim to define a term $ua : \UA_i$. Given arbitrary types $A, B : \U_i$ and an equivalence $(f,e) : \Equiv{A}{B}$ we define $\ua(f,e) : A = B$ as follows:
\begin{align*} 
	\hspace{40pt}A &= \sum\limits_{a : A}\; 1 &&\by{(1)}\\
	 &= \sum\limits_{a : A}\; \sum\limits_{b : B}\; f\; a = b
	   && \by{funext and (3) on $sing (f a)$}\\
	 &= \sum\limits_{b : B}\; \sum\limits_{a : A}\; f\; a = b
	   && \by{(2)}\\
	 &= \sum\limits_{b : B}\; 1
	   && \by{funext and (3) on $fib_f(b)$ (contractible by $e$)}\\ 
	 &= B &&\by{(1)}
\end{align*}
where the proof that $A = B$ is given by the concatenation of each step of the above calculation.

The backwards direction follows from the fact that the obvious maps $A \to \sum_{a : A} 1$ and $(\sum_{a:A}\sum_{b:B} C\; a\; b) \to (\sum_{b:B}\sum_{a:A} C\; a\; b)$ are both easily shown to be bi-invertible and hence equivalences, and from the fact that any contractible type is equivalent to $1$ \cite[Lemma 3.11.3.]{HoTT}. Therefore given $\ua : \UA_i$ we simply apply it to these equivalences to get the required equalities (1)-(3).
\end{proof}

Next, we decompose the computation rule for naive univalence $\UA\beta_i$ into axioms (4) and (5). Since $\UA\beta_i$ depends on $\UA_i$ and axioms (4) and (5) depend on axioms (1) and (2) respectively, we in fact show the logical equivalence between the pair $\UA_i$ and $\UA\beta_i$, and axioms (1)-(5).


\begin{lemma}\label{lem}
	The function $\coerce$ is compositional. That is, given types $A, B, C : \U_i$, and equalities $p : A = B$ and $q : B = C$ we have $\coerce(p \ct q) = \coerce(q) \circ \coerce(p)$.
\end{lemma}
\begin{proof}
	Straightforward by path induction on either of $p$ or $q$, or on both.
\end{proof}

\begin{theorem}\label{main}
	Axioms (1)-(5) for a universe $\U_i$ are together logically equivalent to $\sum_{\ua : \UA_i} \mathit{UA\beta_i}(ua)$.
\end{theorem}
\begin{proof}
For the forwards direction we know from Theorem \ref{ua} that axioms (1) to (3) allow us to construct a term $\ua : \UA_i$. If, in addition, we assume axioms (4) and (5) then we can show that for all $(f,e) : \Equiv{A}{B}$ we have $\coerce(\ua(f,e)) = f$ as follows.
	
Since $\ua$ was constructed as the concatenation of five equalities then, in light of Lemma \ref{lem}, we have that coercing along $\ua(f,e)$ is equal to the result of coercing along each stage of the composite equality $ua(f,e)$. Therefore starting with an arbitrary $a : A$, we can track what happens at each stage of this process like so:
\newcommand{\goesto}{\quad\mapsto\quad}
\[ a \goesto (a,\;*) \goesto (a,\; f\; a,\; \mathit{refl}) \goesto (f\; a,\; a,\; \mathit{refl}) \goesto (f\; a,\; *) \goesto f\; a \] 
Therefore we see that for all $a : A$ we have $\coerce(\ua(f,e))(a) = f(a)$ and hence by function extensionality we have $\coerce(\ua(f,e)) = f$ as required.

For the reverse direction we assume that we are given $\ua : \UA_i$ and $\ua\beta : \UA\beta_i(\ua)$. We can now apply Theorem \ref{ua} to construct terms $\mathit{unit}$, $\mathit{flip}$ and $\mathit{contract}$ satisfying axioms (1) to (3) from $\ua$.

Since $\mathit{unit}$ and $\mathit{flip}$ were constructed by applying $\ua$ to the obvious equivalences, then by $\ua\beta$ we know that applying $\coerce$ to these equalities will return the equivalences that we started with. From this we can easily construct terms $\mathit{unit\beta}$ and $\mathit{flip\beta}$ satisfying axioms (4) and (5) respectively.
\end{proof}

\begin{corollary}\label{cor}
	Axioms (1)-(5) for a universe $\U_i$ are together logically equivalent to the proper univalence axiom for $\U_i$.
\end{corollary}
\begin{proof}
	By combining Theorems \ref{naiveuniv} and \ref{main}.
\end{proof}

%

\section{Applications in models of type theory}\label{sec:cubical}

In this section we discuss one reason why the result given in Corollary \ref{cor} is useful when trying to construct models of univalent type theory. Specifically, we believe that this decomposition is particularly useful for showing that a model of type theory with an interval object (e.g. cubical type theory \cite{CoquandT:cubttc}) supports the univalence axiom. We first explain why we believe this to be the case in general terms, and then give a precise account of what happens in the specific case of the cubical sets model presented in \cite{CoquandT:cubttc}. The arguments given here should translate to many similar models of type theory \cubicalrefsnocubttc.

Note that we are assuming function extensionality. Every model of univalence must satisfy function extensionality \cite[Section 4.9]{HoTT}, but it is often much easier to verify function extensionality than the proper univalence axiom in a model of type theory. In particular, function extensionality will hold in any type theory which includes an appropriate interval type, cf.~\cite[Lemma 6.3.2]{HoTT}.

Experience shows that axioms (1), (2), (4) and (5) are simple to verify in many potential models of univalent type theory. To understand why, it is useful to consider the interpretation of $\Equiv{A}{B}$ in such a model. Propositional equality in the type theory is usually not interpreted as equality in the model's metatheory, but rather as a construction on types e.g.~path spaces in models of HoTT. Therefore, writing $\llbracket X \rrbracket$ for the interpretation of a type $X$, an equivalence in the type theory will give rise to morphisms $f : \llbracket A \rrbracket \to \llbracket B \rrbracket$ and $g : \llbracket B \rrbracket \to \llbracket A \rrbracket$ which are not exact inverses, but rather are inverses modulo the interpretation of propositional equality, e.g.~the existence of paths connecting $x$ and $g(f(x))$, and $y$ and $f(g(y))$ for all $x \in \llbracket X \rrbracket, y \in \llbracket Y \rrbracket$. However, in many models the interpretations of $A$ and $\sum_{a : A} 1$, and of $\sum_{a:A}\sum_{b:B} C\; a\; b$ and $\sum_{b:B}\sum_{a:A} C\; a\; b$ will be isomorphic, i.e. there will be morphisms going back and forth which are inverses up to equality in the model's metatheory. This will be true in any presheaf model of type theory of the kind described in Section \ref{presheafmod}, and should be true more generally in any model which validates eta-rules for $1$ and $\Sigma$.

This means that we can satisfy (1) and (2) by proving that this stronger notion of isomorphism gives rise to a propositional equality between types. Verifying axioms (4) and (5) should then reduce to a fairly straightforward calculation involving two instance of this construction.

This leaves axiom (3), which captures the homotopical condition that every contractible space can be continuously deformed into a point. The hope is that verifying the previous axioms should be fairly straightforward, leaving this as the only non-trivial condition to check.

We now examine what happens in the specific case of cubical sets \cite{CoquandT:cubttc}.


\subsection{Example: the CCHM model of univalent type theory}

In this section we will examine what happens in the case of the Cohen, Coquand, Huber, M\"ortberg (CCHM) model of type theory based on cubical sets \cite{CoquandT:cubttc}. To be clear, this model is shown to validate the univalence axiom in the previously cited paper. However, here we give an alternative, hopefully simpler, proof of univalence using the decomposition given in Section \ref{sec:axi}. We start from the knowledge that cubical sets model a type theory with $Path$ types given by maps out of an interval object $\I$, and where types come equipped with a \emph{composition} operation which is closed under all type formers ($\Sigma, \Pi, Path$). From this we then show how to validate our axioms, and therefore the proper univalence axiom.

For most of this section we will work in the internal language of the cubical sets topos, using a technique developed by the authors in a previous paper \cite{PittsAM:aximct}.
We begin with a brief summary of the cubical sets model and then describe the internal language approach to working with such models. For those unfamiliar with this material we refer the reader to \cite{CoquandT:cubttc} and \cite{PittsAM:aximct} respectively for further details.

\subsubsection{The cubical sets model}
\label{presheafmod}

Cohen~\emph{et al} \cite{CoquandT:cubttc} present a model of type theory using the category $\hat{\mathcal C}$ of presheaves on the small category $\mathcal C$ whose objects are given by finite sets of symbols, written $I, J, K$, with $\mathcal C(I,J)$ being the set of maps $J \to dm(I)$, where $dm(I)$ is the free De Morgan algebra \cite{balbes1977distributive} on the set $I$. 

First we recall the standard way of constructing a presheaf model of type theory \cite{HofmannM:synsdt}, note that this is not the final model construction. Take $\hat{\mathcal{C}}$ to be the category of contexts with the types over a context $\Gamma \in \hat{\mathcal{C}}$, written $Ty(\Gamma)$, given by presheaves on $\Gamma$'s category of elements. Terms of type $A \in Ty(\Gamma)$, written $Ter(\Gamma \vdash A)$ are simply global sections of $A$. Explicitly, this means that a type $A \in Ty(\Gamma)$ is given by a family of sets $A(I,\rho)$ for every $I \in \mathcal C$ and $\rho \in \Gamma(I)$ such that for every $a \in A(I,\rho)$ and $f : J \to I$ we have $A(f)(a) \in A(J,\Gamma(f)(\rho))$ with $A(\id_I)(a) = a$ and $A(g \circ f)(a) = A(g)(A(f)(a))$. A term $a \in Ter(\Gamma \vdash A)$ is given by a family $a(I,\rho) \in A(I,\rho)$ for every $I \in \mathcal C$ and $\rho \in \Gamma(I)$ such that for all $f : J \to I$ we have $A(f)(a(I,\rho)) = a(J,\Gamma(f)(\rho))$. Following the convention in \cite{CoquandT:cubttc} we will often omit the first argument $I$ and will write functorial actions $\Gamma(f)(\rho)$ simply as $\rho f$.

These constructions all model substitution, context extension, projection, etc, in the correct way, and can be shown to form a category with families (CwF) in the sense of Dybjer \cite{DybjerP:intt}. Such a model always supports both dependent sums and products. For example, given types $A \in Ty(\Gamma)$ and $B \in Ty(\Gamma.A)$ then the dependent sum $\Sigma A B \in Ty(\Gamma)$ can be interpreted as
\[ \Sigma A B (I, \rho) \;\triangleq\; \{ (a,b) \mid a \in A(I, \rho),\; b \in B(I, \rho, a) \} \]
As stated in the previous section, any model of this kind will always have $\Sigma A (\Sigma B C)$ being strictly isomorphic to $\Sigma B (\Sigma A C)$, that is, with natural transformations in each direction which are inverses up to equality in the model's metatheory. Furthermore, assuming that the terminal type $1$ is interpreted as the terminal presheaf, then the same will be true for the types $A$ and $\Sigma A 1$. This is potentially useful when verifying axioms (1), (2), (4) and (5) for the reasons given above.

To get a model of type theory which validates the univalence axiom we restrict our attention to types with an associated \emph{composition structure} \cite[Definition 13]{CoquandT:cubttc}. We call such types \emph{fibrant} and write $FTy(\Gamma)$ for the collection of fibrant types over a context $\Gamma \in \hat{\mathcal C}$. Taking contexts, terms, type formers and substitution as before, we get a new CwF of fibrant types. We delay giving the exact definition of a composition structure until after we introduce the internal type theory of $\hat{\mathcal C}$ in the following section.

\subsubsection{The internal type theory of $\hat{\mathcal C}$}

In previous work \cite{PittsAM:aximct} the authors axiomatised the properties of the cubical sets topos needed to develop a model of univalent type theory. They then showed how many of the constructions used in the model could be replicated using the internal type theory of an elementary topos \cite{MaiettiME:modcdt}. Here we will often build on this approach, working mostly in the internal type theory. We now give a brief overview of this approach, and refer the reader to \cite{PittsAM:aximct} for full details.

We use a concrete syntax inspired by
Agda~\cite{Agda}. Dependent function types are written as
$(x:A)\fun B$ with lambda abstractions
written as $\lambda(x:A)\fun t$. We use $\{\}$ in place of $()$ to indicate the use of implicit arguments. Dependent product types are written
as $(x:A)\times B$ with the pairing operation written as
$(s,t)$.

We assume the existence of an interval object $\I$ with endpoints
$\src, \tgt : 1 \to \I$ subject to certain conditions, a class of
propositions $\Cof \mono \Omega$, closed under $\vee, \wedge$ and
$\I$-indexed $\forall$, which we call the \emph{cofibrant}
propositions and an internal full subtopos $\Univ$. Given
$\varphi : \Prop$ we write
$[ \varphi ] \defeq \{ \_ : 1 \mid \varphi \}$ for the type whose
inhabitation corresponds to the provability of $\varphi$, and given a
object $\Gamma : \Univ$ and a cofibrant property
$\Phi : \Gamma \fun \Cof$ we wrtie
$\Gamma\resby \Phi \defeq (x:\Gamma)\times[\Phi\,x]$ for the
\emph{restriction of $\Gamma$ by $\Phi$}. Given $\varphi : \Cof$, $f:[\varphi]\fun A$
and $a:A$ we write $(\varphi,f)\exto a$ for $(u:[\varphi])\fun f\,u = a$; thus
elements of this type are proofs that the partial element $f$ (with
cofibrant domain of definition $\varphi$) extends to the totally
defined element $a$.

As an example of the use of this language we now reproduce the internal definition of a fibration \cite[Definition 5.7]{PittsAM:aximct}:

\begin{definition}[CCHM fibrations]%
\label{def:cchm-fib}
  A \emph{CCHM fibration} $(A,\alpha)$ over a type $\Gamma:\Univ$ is a
  family $A:\Gamma\fun \Univ$ equipped with a fibration structure
  $\alpha : \isFib A$, where
  $\isFib : \{\Gamma:\Univ\}(A: \Gamma\fun\Univ) \fun \Univ$ is
  defined by
  \begin{align*}
    \isFib\,\{\Gamma\}\,A \defeq (e:\Bool)(p:\I\fun \Gamma) \fun \Comp
    e\, (A\comp p) 
  \end{align*}
  Here ${\Comp} : (e:\Bool)(A:\I\fun\Univ) \fun\Univ$ is the type
  of \emph{composition structures} for $\I$-indexed families:
  \begin{align*}
    \Comp e\, A \defeq 
    \begin{array}[t]{@{}l}
      (\varphi:\Cof)(f: [\varphi]\fun\Pi_\I A)\fun{}\\
      \{a_0:A\,e \mid (\varphi,f)\conc e \exto a_0\} \fun
      \{a_1: A\,\inv{e} \mid (\varphi,f)\conc\inv{e} \exto a_1\}  
    \end{array}
  \end{align*}
  where $(\varphi,f)\conc e$ is an abbreviation for the term $\lambda
  u:[\varphi].\,f\,u\,e$ of type  $[\varphi ]\fun A\,e$.
  
\end{definition}

We write $\Fib\, \Gamma = (A : \Gamma \fun \Univ) \times \isFib\, A$ for the type of fibrations over $\Gamma : \Univ$ and recall that fibrations can be reindexed $(A, \alpha)[\gamma] = (A \circ \gamma, \alpha[\gamma]) : \Fib\, \Delta$ for $\gamma : \Delta \fun \Gamma$.

\subsubsection{Paths between fibrations}

We work in the internal type theory of the cubical sets topos wherever possible, however this approach does have its limitations. In particular this internal approach is unable to describe type theoretic universes \cite[Remark 7.5]{PittsAM:aximct}. Therefore we will not be able to construct elements of the identity type on the universe (the target type of axioms (1)-(3)). Instead, we work with an (externally) equivalent notion of equality between types.

\begin{definition}[Path equality between fibrations]
  \label{defi:patebf}
  Define the type of paths between CCHM fibrations
  $\IdU{\_}{\_} : \{\Gamma:\Univ\} \fun \Fib\Gamma \fun \Fib\Gamma
  \fun \Univ_1$ by
  \[
  \IdU{(A,\alpha)}{(B,\beta)} \defeq \{ (P,\rho) : \Fib(\Gamma\times\I) \mid 
  (P,\rho)[\langle\id,\src\rangle] = (A,\alpha) \wedge
  (P,\rho)[\langle\id,\tgt\rangle] = (B,\beta)\} 
  \]
\end{definition}

To understand why this notion of path is equivalent to the usual
notion recall that the universe construction in \cite{CoquandT:cubttc} is given by the usual Hofmann-Streicher universe construction for presheaf categories~\cite{HofmannM:lifgu}. This means that there exists a type $\mathcal U \in Ty(\Gamma)$ for all $\Gamma$ given by $\mathcal U(I,\rho) \defeq FTy(\y I)$ where $\y I$ denotes the Yoneda embedding of $I$. Every (small) fibrant type $A \in FTy(\Gamma)$ has a code $\ulcorner A \urcorner \in Ter(\Gamma \vdash \mathcal U)$ and every $a \in Ter(\Gamma \vdash \mathcal U)$ encodes a type $El\,a \in FTy(\Gamma)$ such that $El\, (\ulcorner A \urcorner) = A$ and $\ulcorner El\, a \urcorner = a$ for all $a$ and $A$.

Now consider the following: externally, a path $P : \IdU{A}{B}$
corresponds to a fibration $P \in FTy(\Gamma.\mathbb{I})$ such that
$P[\langle\id,0\rangle] = A$ and $P[\langle\id,1\rangle] = B$
for some $\Gamma \in \hat{\mathcal C}$ and $A,B \in FTy(\Gamma)$.
From this data we can construct $p \in Ter(\Gamma \vdash
Path\; \Univ\; \ulcorner A \urcorner\; \ulcorner B \urcorner)$
like so:
  \[ p(\rho) \defeq \langle i \rangle\; \ulcorner P \urcorner(\rho s_i, i) \]
for $I \in \mathcal C, \rho \in \Gamma(I)$. Note that this does define a
path with the correct endpoints since substituting 0 for i we get:
  \[ (\ulcorner P \urcorner(\rho, 0))f =
        P (\rho f, 0 f) =
        P (\rho f, 0) =
        A (\rho f) =
        (\ulcorner A \urcorner (\rho))f \]
for all $f : J \to I$. The case for $i = 1$ is similar.

Conversely, given $p \in Ter(\Gamma \vdash
Path\; \Univ\; \ulcorner A \urcorner\; \ulcorner B \urcorner)$
we can define $P \in FTy(\Gamma.\mathbb{I})$ with the required properties
like so:
  \[ P(\rho, i) \defeq (p\; \rho\; i ) \id_I  \]
for $I \in \mathcal C, \rho \in \Gamma(I), i \in \mathbb{I}$. Again, note
that this has the correct properties, e.g.~at 0:
  \[ P[\langle\id,0\rangle]\;\rho =
        P (\rho, 0) =
        (p\; \rho\; 0) \; \id_I  =
        (\ulcorner A \urcorner\; \rho) \id_I =
        (El \ulcorner A \urcorner)\rho =
        A\rho \]
for all $\rho \in \Gamma(I)$. It is easily checked that these two
constructions are mutual inverses. Therefore the data described by
$\IdU{\_}{\_}$ corresponds exactly to the data required to describe a path
in the universe.

\subsubsection{The realignment lemma}

Next, we introduce a technical lemma that will be needed in the following sections.  For readers familiar with the cubical sets model, it is interesting to note that this is the only place where we use the fact that
cofibrant propositions are closed under $\I$-indexed $\forall$ \cite[Section 4.1]{CoquandT:cubttc}.

\begin{lemma}[Realignment lemma]\label{realign}
  Given $\Gamma:\Univ$ and $\Phi:\Gamma\fun\Cof$, let
  $\iota:\Gamma\resby\Phi \mono \Gamma$ be the first projection.
  For any $A:\Gamma\fun \Univ$,
  $\beta : \isFib(A \comp \iota)$ and $\alpha : \isFib{A}$, there
  exists a composition structure $\kw{realign}(\Phi,\beta,\alpha) : \isFib{A}$ such that
  $\beta = \kw{realign}(\Phi,\beta,\alpha)[\iota]$.
\end{lemma}
\begin{proof} By \cite[Theorem 6.13]{PittsAM:aximct}, in which we define $\kw{realign}(\Phi,\beta,\alpha)$ by
  \begin{equation}
    \kw{realign}(\Phi,\beta,\alpha)\,e\,p\,\psi\,f\,a \defeq
    \alpha\,e\,p\,(\psi\disj(\forall(i:\I).\,\Phi\,(p\,i)))\,(f
    \cup f')\,a
  \end{equation}
  where $f' : [\forall(i:\I).\,\Phi\,(p\,i)] \fun \Pi_\I (A \circ p)$ is given
  by
  $f'\,u \defeq \fil \,e\,\beta\,(\lambda i \fun (p,
  u\,i))\,\psi\,f\,a$.
\end{proof}

In words, given a fibrant type $A$ and an alternative composition structure defined only on some restriction of $A$, then we can \emph{realign} the original composition structure so that it agrees with the alternative on that restriction.

Note that this construction is stable under reindexing in the following sense: given $\gamma : \Delta \to \Gamma$, $\Phi:\Gamma\fun\Cof$, $A:\Gamma\fun \Univ$,
  $\beta : \isFib(A \comp \iota)$ and $\alpha : \isFib{A}$ then,
  \begin{align*}
  	\kw{realign}&(\Phi,\beta,\alpha)[\gamma]\,e\,p\,\psi\,f\,a\\
  	  &= \kw{realign}(\Phi,\beta,\alpha)\,e\,(\gamma \circ p)\,\psi\,f\,a\\
  	  &= \alpha\,e\,(\gamma \circ p)\,(\psi\disj(\forall(i:\I).\,\Phi\,((\gamma \circ p)\,i)))\,(f
    \cup \fil \,e\,\beta\,(\lambda i \fun (\gamma \circ p,
  u\,i))\,\psi\,f\,a)\,a\\
  	  &= \alpha[\gamma]\,e\,p\,(\psi\disj(\forall(i:\I).\,(\Phi \circ \gamma)(p\,i)))\,(f
    \cup \fil \,e\,\beta[\langle \gamma, \id \rangle]\,(\lambda i \fun (p,
  u\,i))\,\psi\,f\,a)\,a\\
  	  &= \kw{realign}(\Phi \circ \gamma,\beta[\langle \gamma , \id \rangle],\alpha[\gamma])\,e\,p\,\psi\,f\,a
  \end{align*}
Therefore we have
\[\kw{realign}(\Phi,\beta,\alpha)[\gamma] = \kw{realign}(\Phi \circ \gamma,\beta[\langle \gamma , \id \rangle],\alpha[\gamma]) \]

\subsubsection{Fibrations are closed under isomorphism}

\begin{definition}[Strict isomorphism]\label{strictiso}
	A \emph{strict isomorphism} between two objects $A,B : \Univ$ is a pair $(f,g)$ where $f : A \to B$ and $g : B \to A$ such that $g \circ f = \id$ and $f \circ g = id$. We write $(f,g) : A \cong B$.
	
	This notion lifts to both families and fibrations, and we overload the notation $\_ \cong \_$ like so: when $A,B : \Gamma \fun \Univ$ then we take $A \cong B$ to mean $(x : \Gamma) \to A\,x \cong B\,x$ and when $A,B : \Fib\,\Gamma$ then we take $A \cong B$ to mean $\fst A \cong \fst B$.
\end{definition}

\begin{lemma}\label{isofib}
	Given a family $A : \Gamma \to \Univ$ and a fibration $(B , \beta) : \Fib\,\Gamma$, such that $A \cong B$, then we can construct an $\alpha$ such that $(A,\alpha) : \Fib\,\Gamma$.
\end{lemma}
\begin{proof}
	Assume that we are given $A$ and $(B,\beta)$ as above with an isomorphism $\langle f,g\rangle : A \cong B$. We can then define a composition structure for $A$ as follows:
	\[ \alpha\; e\; p\; \varphi\; q\; a_0 \defeq g\; (p\; \inv{e})\; (\beta\; e\; p\; \varphi\; (\lambda u\, i \fun f\,(p\, i)\, (q\; u\; i))\; (f\, (p\, e)\, a_0) ) \]
	This construction has the required property that, given $u : [ \varphi ]$:
\begin{align*} 
	\alpha\; e\; p\; \varphi\; q\; a_0
	  &= g\; (p\; \inv{e})\; (\beta\; e\; p\; \varphi\; (\lambda u\, i \fun f\,(p\, i)\, (q\; u\; i))\; (f\, (p\, e)\, a_0) )\\
	  &= g\; (p\; \inv{e})\; (f\, (p\; \inv{e})\, (q\; u\; \inv{e}))\\
	  &= q\; u\; \inv{e}
\end{align*}
Hence $(A,\alpha) : \Fib\,\Gamma$.
\end{proof}
Note that this proof only uses the fact that $g\, x \circ f\, x = \id$ and so in fact the lemma holds more generally in the case where $\langle f , g \rangle$ is just a section-retraction pair rather than a full isomorphism. Although here will we only use it in the context of isomorphisms.

\subsubsection{Strictification}

\begin{theorem}\label{strictify-pre}
	There exists a term:
	\[
  \kw{strictify} :
  \begin{array}[t]{@{}l}
    \{ \varphi : \Cof \}
    (A: [\varphi] \fun \Univ)
    (B: \Univ)
    (s : (u : [\varphi]) \fun (A\, u \cong B))
    {}\fun{}\\
    (B' : \Univ) \times \{s' : B' \cong B \mid \forall(u :
    [\varphi]).\; A\, u = B' \wedge s\,u = s'\}
  \end{array}
  \]
  In words, this says that given any object $B : \Univ$ and any cofibrant partial object $A : [\varphi] \fun \Univ$ such that $A$ is isomorphic to $B$ everywhere it is defined, then one can can construct a new object $B' : \Univ$ which extends $A$, is isomorphic to $B$, and this isomorphism extends the original isomorphism.
\end{theorem}
\begin{proof}
  See \cite[Theorem 8.4]{PittsAM:aximct} for a proof that this
  property holds in the cubical sets model, and more generally in many
  other models (classically, in all presheaf models).
\end{proof}

We now lift this strictification property from objects to fibrations.

\begin{theorem}\label{strictify}
	Given $\Gamma : \Univ$ and $\Phi : \Gamma \to \Cof$, a partial fibration $A : \Fib(\Gamma \resby \Phi)$ and a total fibration $B : \Fib\,\Gamma$ with $iso : A \cong B[\iota]$, we can construct a new type and isomorphism:
	\[A' : \Fib(\Gamma) \qquad \text{ and } \qquad iso'\; :\; A' \cong B \]
such that
\[ A'[\iota] = A \qquad\text{ and }\qquad iso' \circ \iota = iso \]
where $\iota$ is the inclusion $\Gamma\resby\Phi \mono \Gamma$.
\end{theorem}
\begin{proof}
	Given $\Gamma : \Univ$, $\Phi : \Gamma \to \Cof$, $(A,\alpha) : \Fib(\Gamma \resby \Phi	)$ and $(B,\beta) : \Fib\,\Gamma$ with $iso : A \cong B \circ \iota$, we define $A'$, $iso'$ as:
\[A'\, x \defeq \fst(\kw{strictify}(A(x,\_), B\,x , iso(x,\_))) \]
\[ iso'\, x \defeq \snd(\kw{strictify}(A(x,\_), B\,x , iso(x,\_))) \]
Now consider the equalities that are required to hold. From the properties of $\kw{strictify}$ we already have that $A' \circ \iota = A$ and $iso' \circ \iota = iso$. Therefore we just need to define a composition structure $\alpha' : \isFib\,A'$ such that $\alpha'[\iota] = \alpha$.

Since $A' \cong B$ and $\beta : \isFib\; B$ we can use Lemma \ref{isofib} to deduce that $A'$ has a composition structure, which we call $\alpha'_{pre}$. We then define $\alpha' \defeq \kw{realign}(\Phi,\alpha,\alpha'_{pre})$ using Lemma \ref{realign}.

\end{proof}

\subsubsection{Misaligned paths between fibrations}

We now introduce an new relation between fibrations which we call a misaligned path. This is similar to the notion of path between fibrations introduced in Definition \ref{defi:patebf}, except that rather than being equal to $A$ and $B$ at the endpoints, the path only need be isomorphic to $A$ and $B$ at the endpoints.

\begin{definition}[Misaligned path equality between fibrations]
  \label{defi:patebf-mis}
  Define the type of misaligned paths between CCHM fibrations
  $\IdCong{\_}{\_} : \{\Gamma:\Univ\} \fun \Fib\Gamma \fun \Fib\Gamma
  \fun \Univ_1$ by
  \[
  \IdCong{(A,\alpha)}{(B,\beta)} \defeq ((P,\rho) : \Fib(\Gamma\times\I)) \times 
  (A \cong P \circ \langle\id,\src\rangle) \times
    (B \cong P \circ \langle\id,\tgt\rangle) 
  \]
\end{definition}

We can show that every misaligned path can be improved to a regular path between fibrations. First, we introduce a new construction on fibrations.

\begin{definition}
	Given fibrations $A, B : \Fib\,\Gamma$ we define a new fibration
		\[ A \veebar B \;:\; \Fib((\Gamma \times \I) \resby \Phi) \qquad\text{ where }\qquad \Phi(x,i) \defeq (i = \src) \vee (i = \tgt) \]
	given by $(A , \alpha) \veebar (B , \beta) \defeq (C, \gamma)$ where
		\begin{align*}
		  C &: (\Gamma\times\I) \resby \Phi \fun \Univ\\
		  C&\,((x,i),u) \defeq ((\lambda \_ : [i=\src] \fun A\,x) \join (\lambda
		     \_ : [i=\tgt] 
		     \fun B\,x))\, u
		\end{align*}
	Here $C$ is a sort of disjoint union of the families $A$ and $B$, observing that $(\Gamma \times \I)\resby \Phi \cong \Gamma + \Gamma$ then we can think of $C$ as essentially being $[A,B] : \Gamma + \Gamma \to \Univ$.
	
	To see that $C$ is fibrant we observe that the interval $\I$ is internally connected in the sense of $\pos{1}$ in \cite[Figure 1]{PittsAM:aximct}. This means that any path $p : \I \fun (\Gamma \times \I)\resby\Phi$ must either factor as $p = \langle p' , \src, * \rangle$ or as $p = \langle p' , \tgt, *\rangle$. Therefore any composition problem for $C$ must lie either entirely in $A$, in which case we use $\alpha$ to construct a solution, or entirely in $B$, in which case we use $\beta$. For further detail we refer the reader to \cite[Theorem 7.3]{PittsAM:aximct} where the family $C$ occurs as an intermediate construction.
\end{definition}
\begin{definition}
		Given $\Gamma:\Univ$, $A, B : \Fib\,\Gamma$ and $D : \Fib(\Gamma\times\I)$ with isomorphisms $iso_{\src} : A \cong D[\langle\id,\src\rangle]$ and $iso_{\tgt} : B \cong D[\langle\id,\tgt\rangle]$ then define $iso_{\src} \veebar iso_{\tgt} : A \veebar B \cong D[\iota]$ as follows. Given $(x,i,u) : (\Gamma \times \I) \resby \Phi$:
\begin{align*}
	(iso_{\src} \veebar iso_{\tgt})\, (x,i,u) &: (\fst (A \veebar B))\; (x,i,u) \to (\fst D)\;(x,i)\\
	(iso_{\src} \veebar iso_{\tgt})\, (x,i,u) &\defeq \left\{\begin{matrix}
 iso_{\src}\,x & \text{when } u : [ i = \src ] \\ 
 iso_{\tgt}\,x & \text{when } u : [ i = \tgt ]
\end{matrix}\right.
\end{align*}
\end{definition}
Observe that for all $A, B : \Fib\,\Gamma$ we have $(A \veebar B)[\langle \id , \src, * \rangle] = A$ and $(A \veebar B)[\langle \id , \tgt, * \rangle] = B$, and for all $iso_{\src} : A \cong D[\langle\id,\src\rangle]$ and $iso_{\tgt} : B \cong D[\langle\id,\tgt\rangle]$ we have $(iso_{\src} \veebar iso_{\tgt}) \circ \langle \id , \src, * \rangle = iso_{\src}$ and $(iso_{\src} \veebar iso_{\tgt}) \circ \langle \id , \tgt, * \rangle = iso_{\tgt}$. We now use this construct to show the following result:

\begin{lemma}\label{improve}
	There exists a function
	  \[ \kw{improve} : \{\Gamma:\Univ\}\{A\; B : \Fib\Gamma\} \fun \IdCong{A}{B} \fun \IdU{A}{B} \]
\end{lemma}
\begin{proof}
	Take $\Gamma:\Univ$, $A, B : \Fib\Gamma$ and $(P,iso_{\src},iso_{\tgt}) : \IdCong{A}{B}$ and observe that
	$iso_{\src} \veebar iso_{\tgt} : A \veebar B \cong P[\iota]$. Therefore we can use Theorem \ref{strictify} to strictify $P$ in order to get $P' : \Fib(\Gamma \times \I)$ such that $P'[\iota] = A \veebar B$, where $\iota$ is the restriction $(\Gamma\times\I)\resby\Phi \fun \Gamma \times \I$. Now consider reindexing $P'$ along $\langle \id, \src \rangle : \Gamma \to \Gamma \times \I$ we get:
\[ P'[\langle \id, \src \rangle]
     = P'[\iota \circ \langle \id, \src, * \rangle]
     = P'[\iota][\langle \id, \src, * \rangle]
     = (A \veebar B)[\langle \id, \src, * \rangle]
     = A
     \]
and similarly $P'[\langle \id, \tgt \rangle] = B$. Therefore we have $P' : \IdU{A}{B}$ as required.
\end{proof}

\subsubsection{Function extensionality}

As discussed previously, function extensionality holds straightforwardly in any type theory which includes an interval object/type with certain computational properties, cf.~\cite[Lemma 6.3.2]{HoTT}. See \cite[Remark 5.16]{PittsAM:aximct} and \cite[Section 3.2]{CoquandT:cubttc} for a proof in the case of cubical type theory.

\subsubsection{Axioms (1), (2), (4) and (5)}

As discussed previously, we can satisfy axioms (1) and (2) by showing that there is a way to construct paths between strictly isomorphic (fibrant) types $A,B : \Fib\,\Gamma$.

\begin{theorem}\label{isovalence}
	Given fibrations $A , B : \Fib\,\Gamma$ with $iso : A \cong B$ we can construct a path $\kw{isopath}(iso) : \IdU{A}{B}$.
\end{theorem}
\begin{proof}
Given $A,B,f,g$ as above, let $B' \defeq B[\fst] : \Fib(\Gamma \times \I)$ and note that $iso : A \cong B'[\langle \id, \src \rangle ]$ and $\id : B \cong B'[\langle \id, \tgt \rangle]$ where $\id$ is the obvious isomorphism $B \cong B$. Therefore we can define
  \[ \kw{isopath}(iso) \defeq \kw{improve}(B[\fst], iso, \id) : \IdU{A}{B} \]
  as required. Note that, in this case, $\kw{improve}$ will in fact only improve $B[\fst]$ at $\src$, since at $\tgt$ we improve along the identity, which does nothing.
\end{proof}

\begin{corollary}\label{cor12cubsets}
	Axioms (1) and (2) hold in the cubical sets model.
\end{corollary}
\begin{proof}
	The obvious isomorphisms $A \cong A \times 1$ and $\sum_{a:A}\sum_{b:B} C\; a\; b \cong \sum_{b:B}\sum_{a:A} C\; a\; b$ are both clearly strict isomorphisms in the sense of Definition \ref{strictiso}. Therefore we can construct the required paths $\IdU{A}{(A \times 1)}$ and $\IdU{(\sum_{a:A}\sum_{b:B} C\; a\; b)}{(\sum_{b:B}\sum_{a:A} C\; a\; b)}$. Hence axioms (1) and (2) hold. 
	
	Note that in order interpret axioms (1) and (2) using $\kw{isopath}$ we need to know that $\kw{isopath}$ is stable under reindexing (substitution in the type theory). This will be the case because most of the constructions used to define it (strictification, closure under isomorphism, etc) are all performed fiberwise and hence will be stable under reindexing. The only exception is the realignment lemma, which redefines the entire composition structure. However, we previously showed $\kw{realign}$ to be stable under reindexing. Therefore $\kw{isopath}$ will also be stable under reindexing.
\end{proof}

We have seen that we can easily satisfy axioms (1) and (2) in the cubical sets model. However, we also need to know what happens when we coerce along these equalities. This can be stated in general for any strictly isomorphic types.

\begin{theorem}\label{coerceiso}
	Given fibrations $(A,\alpha),(B,\beta) : \Fib\,\Gamma$ with $\langle f, g \rangle : A \cong B$, coercing along $\kw{isopath}(\langle f, g \rangle)$ is (propositionally) equal to applying $f$.
\end{theorem}
\begin{proof}
	Take $(A,\alpha),(B,\beta), f, g$ as above and let $(P, \rho) = \kw{isopath}(\langle f,g \rangle )$. By unfolding the constructions used we can see that $\rho$ was obtained by realigning some $\rho_{pre}$, which in turn was obtained by transferring $\beta[\fst]$ across the isomorphism:
	\[\mathit{iso}'\,(x,i) = \snd (\kw{strictify}((A,\beta) \veebar (B,\beta)(x,i,\_), B\,x,(\langle f, g \rangle\veebar\id)\,(x,i,\_))) : P\,x \cong B\, x\]
	Now consider arbitrary $x : \Gamma, a_0 : A\; x$ and note that
	\[ \mathit{iso}'\,(x,\src) = (\langle f, g \rangle \veebar \id)\,(x,\src) = \langle f, g \rangle\,x = (f\,x,g\,x) \]
	and
	\[ \mathit{iso}'\,(x,\tgt) = (\langle f, g \rangle \veebar \id)\,(x,\tgt) = (\id, \id) \]
	Now calculate:
\begin{align*} 
	\null\hspace{-10pt}&\kw{coerce}\; \kw{isopath}(\langle f, g \rangle)\; x\; a_0 \\
	 	 \null\hspace{-10pt}&\quad= \rho\; \src\; \langle x , \id \rangle\; \False\; \elim_{\emptyset}\; a_0  &&\by{unfolding definitions}\footnotemark \\
	 	 \null\hspace{-10pt}&\quad = \rho_{pre}\; \src\; \langle x , \id \rangle\; (\forall i.(i=\src \vee i=\tgt))\; q\; a_0  &&\by{Lemma \ref{realign} (for some $q$)} \\
	 	 \null\hspace{-10pt}&\quad = \rho_{pre}\; \src\; \langle x , \id \rangle\; \False\; \elim_{\emptyset}\; a_0  &&\by{definition of $\forall$} \\
	 	 \null\hspace{-10pt}&\quad= \snd (\mathit{iso}'(x,\tgt))\;(\beta\; \src\; \langle x , \id \rangle\; \False\; \elim_{\emptyset}\; (\fst (\mathit{iso}'(x,\src))\; a_0))  &&\by{Lemma \ref{isofib}} \\
	 	 \null\hspace{-10pt}&\quad= \beta\; \src\; \langle x , \id \rangle\; \False\; \elim_{\emptyset}\; (\fst (\mathit{iso}'(x,\src))\; a_0))  && \text{since $\snd (\mathit{iso}'(x,\tgt)) = id$} \\
	 	 \null\hspace{-10pt}&\quad= \beta\; \src\; \langle x , \id \rangle\; \False\; \elim_{\emptyset}\; (f\; x\; a_0) && \text{since $\fst (\mathit{iso}'(x,\src)) = f\,x$} 
\end{align*}
Since this is merely a trivial/empty composition applied to $f\;x\;a_0$ we can construct a path from $f\;x\; a_0$ to $\coerce\; \kw{isopath}(\langle f, g \rangle)\; x\; a_0$ like so:
\[ \fil\; \src\; \beta\; \langle x , \id \rangle\; \False\; \elim_{\emptyset}\; (f\; x\; a_0) : f\; x\; a_0 \path \coerce\; \kw{isopath}(\langle f, g \rangle)\; x\; a_0 \]
Therefore, coercing along $\mathit{isopath}(\langle f, g \rangle)$ is always propositionally equal to applying $f$.
\end{proof}

\footnotetext{Note that there are different ways to interpret $\coerce$ in the model. This interpretation is not in general the same as the one obtained by directly interpreting Definition \ref{def:coerce}. However, the two interpretations will always be path equal in the model (the other interpretation will have more trivial/empty compositions), and so the result still holds when using the other interpretation.}

\begin{corollary}
	Axioms (4) and (5) hold in the cubical sets model (for the terms constructed in Corollary \ref{cor12cubsets}).
\end{corollary}
\begin{proof}
	By Theorem \ref{coerceiso}.
\end{proof}


\subsubsection{Axiom (3)}

In light of the previous section, the only axiom remaining is axiom (3). 
Our goal here is, given a contractible fibration $A : \Fib\,\Gamma$, to define a path $\IdU{A}{1}$. Note that, for any $\Gamma : \Univ$, there exists a unique fibration structure $!_1$ such that $(\lambda \_ \to 1 , !_1) : \Fib(\Gamma)$. Therefore we will ambiguously write $1 : \Fib(\Gamma)$ for the pair $(\lambda \_ \to 1 , !_1)$.

\begin{definition}[The contraction of a family]\label{ctype}
	Given a family $A : \Gamma \to \Univ$ we define the contraction of $A$ as
\begin{align*}
	C_A & : \Gamma \times \I \to \Univ\\
	C_A &(x,i) \defeq [ i = \src ] \to A(x) 
\end{align*}
\end{definition}

We now need to show that $C_A$ is fibrant when $A$ is both fibrant and contractible. First, we restate the property of being contractible (Definition \ref{def:contr}) in the internal type theory.

\begin{definition}
	  A type $A$ is said to be \emph{contractible} if it has a centre of
  contraction $a_0 : A$ and every element $a : A$ is propositionally
  equal to $a_0$, that is, there exists a path $a_0 \path a$.
  Therefore a type is contractible if $\Contr A$ is inhabited, where
  ${\Contr} : \Univ \fun \Univ$ is defined by
  \[ 
  \Contr A \defeq (a_0 :A)\times((a:A)\fun {a_0\path a}) 
  \]
  We say that a family
  $A:\Gamma\fun\Univ$ is contractible if each of its fibres is and
  abusively write
  \[ 
  \Contr A \defeq (x:\Gamma)\fun \Contr(A\,x) 
  \]
\end{definition}
Next we recall the notion of an extension structure \cite[Definition 6.4]{PittsAM:aximct}.
\begin{definition}[Extension structures]\label{defi:ext}
  The type of extension structures, ${\Ext} : \Univ \fun \Univ$, is
  given by
  \[
  \Ext A \defeq (\varphi : \Cof)(f : [ \varphi ] \fun A) \fun \{a:A\mid (\varphi,f)\exto a\}
  \]
  Having an extension structure for a type $A:\Univ$ allows us to extend
  any partial element of $A$ to a total element. As before we say that a family
  $A:\Gamma\fun\Univ$ has an extension structure if each of its fibres
  do, and write
  \[ 
  \Ext A \defeq (x : \Gamma) \fun \Ext (A\, x)
  \]
\end{definition}
\begin{lemma}
	Any family $A : \Gamma \to \Univ$ that is both fibrant and contractible is also extendable in the sense of Defintion \ref{defi:ext}.
\end{lemma}
\begin{proof}
	By \cite[Lemma 6.6]{PittsAM:aximct}.
\end{proof}
Now we can construct a fibrancy structure for $C_A$ as follows:
\begin{theorem}\label{cfib}
	If $(A,\alpha) : \Fib\,\Gamma$ is contractible then we can construct a composition structure for $C_A$.
\end{theorem}
\begin{proof}
Take $(A,\alpha) : \Fib\,\Gamma$ as above. Since $A$ is both fibrant and contractible then we can construct an extension structure $\epsilon : \Ext A$. We can then define a composition structure $c_{\alpha} : \isFib(C_A)$ like so:
\[ (c_\alpha\; e\; p\; \varphi\; f\; c_0)\; u \defeq \epsilon\; \varphi\; (\lambda v \to f\, v\, \inv{e}\, u) \]
for $u : [\snd (p\; \inv{e}) = \src]$. Given $v : [ \varphi ]$ we have:
\[ c_\alpha\; e\; p\; \varphi\; f\; c_0 = \lambda u \to \epsilon\; \varphi\; (\lambda v \to f\, v\, \inv{e}\, u)
  = \lambda u \to f\, v\, \inv{e}\, u = f\, v\, \inv{e} \]
 as required. Therefore we have a defined a valid composition operation for $C_A$.
\end{proof}


\begin{theorem}
	There exists a function
	  \[ \kw{contract} : \{\Gamma : \Univ\}(A : \Fib\,\Gamma) \fun \Contr\;A \fun \IdU{A}{1} \]
\end{theorem}
\begin{proof}
Given $\Gamma : \Univ$, $(A , \alpha) : \Fib\;\Gamma$ and $\epsilon : \Contr\;A$, we obverse that
		\[C_A[\langle \id, \src \rangle](x) = C_A(x,\src) = [ \src = \src ] \fun A(x) \cong 1 \fun A(x) \cong A(x) \]
	and
		\[C_A[\langle \id, \tgt \rangle](x) = C_A(x,\tgt) = [ \tgt = \src ] \fun A(x) \cong \emptyset \fun A(x) \cong 1 \]
Therefore we have $((C_A,c_\alpha),iso_A,iso_1) : \IdCong{A}{1}$ where $iso_A : A \cong C_A[\langle \id, \src \rangle]$ and $iso_1 : 1 \cong C_A[\langle \id, \tgt \rangle]$ are the obvious isomorphisms indicated above.
Hence we can define
\[\kw{contract}((A,\alpha),\epsilon) \defeq \kw{improve}((C_A,c_\alpha),iso_A,iso_1) : \IdU{(A,\alpha)}{1}\]
as required.
\end{proof}

\begin{corollary}
	Cubical type theory with the cubical sets model supports axiom (3).
\end{corollary}
As in Corollary \ref{cor12cubsets} we need to check that $\kw{contract}$ is stable under reindexing (substitution). This holds for the same reasons as before, namely that the only non fibrewise construction used in the definition of $\kw{contract}$ is $\kw{realign}$ which we previously showed to be stable under reindexing.

\section{An application to an open problem in type theory}\label{sec:open}


In Section \ref{sec:prelim} we defined $\funext$ to be the principle which says that two functions $f, g : \prod_{x : A} B(x)$ are equal if they are pointwise equal: $f \sim g \defeq \prod_{x : A} f(x) = g(x)$. That is, we assumed the existence of a term:
\begin{align*}
	\funext_{i,j}:&\prod_{A : \U_i}\;\prod_{B : A \to \U_j }\;\prod_{f, g : \Pi_{x : A}B(x)} f \sim g \to f = g
\end{align*}
for all universe levels $i, j$. This is similar to the statement of naive univalence, $\UA$, from Definition \ref{def:naive-univ} and we call this principle naive function extensionality.

As with proper univalence (Definition \ref{def:proper-univ}), we could have instead stated that the canonical map $\happly : (f = g) \to f \sim g$ is an equivalence. In fact, these two formulations turn out to be equivalent.

\begin{theorem}[due to Voevodsky]\label{funext}
Naive function extensionality is logically equivalent to the proper function extensionality axiom. That is, the existence of a term:
\begin{align*}
	\funext_{i,j}:&\prod_{A : \U_i}\;\prod_{B : A \to \U_j }\;\prod_{f, g : \Pi_{x : A}B(x)} f \sim g \to f = g
\end{align*}
is logically equivalent to the statement that, for all types $A :\U_i$, $B : A \to \U_j$ and maps $f, g : \prod_{x : A} B(x)$, the map $\happly : (f = g) \to (f \sim g)$ is an equivalence.
\end{theorem}
\begin{proof}
For the forwards direction: assuming $\funext$ as above, it is easy to derive a proof of \emph{weak function extensionality} \cite[Definition 4.9.1]{HoTT}. This in turn implies the proper function extensionality axiom by \cite[Theorem 4.9.5]{HoTT}. The reverse direction follows trivially.
\end{proof}

Compare this result with Theorem \ref{naiveuniv} where we saw that naive univalence with a computation rule is logically equivalent to the proper univalence axiom. In the case of function extensionality we did not need to assume any sort of computation rule about $\funext$. Therefore an obvious question is whether this computation rule is in fact necessary in the case of univalence, or whether, as is the case with function extensionality, it is in fact redundant.

\begin{conjecture}\label{prop}
Naive univalence implies the proper univalence axiom. That is, given $\UA_i$,
it follows that for all types $A, B : \U_i$ the map $\idtoeqv : (A = B) \to (\Equiv{A}{B})$ is an equivalence.
\end{conjecture}

To the authors' best knowledge the status of Conjecture \ref{prop} is currently unknown. It is certainly not inconsistent since there are models where naive univalence fails to hold, such as the $\mathit{Set}$-valued model \cite{HofmannM:synsdt}, and models where full univalence holds, such as the cubical sets model \cite{CoquandT:cubttc}. However it is not clear whether Conjecture \ref{prop} is either a theorem of type theory, cf.~the case with function extensionality, or whether there are models which validate $\UA$ but which do not validate the proper univalence axiom.

The work presented here may offer an approach to tackling this problem, by reducing it to the following:

\begin{conjecture}\label{conj2}
Axioms (1)-(3) imply axioms (4)-(5), for possibly modified $\mathit{unit}$ and $\mathit{flip}$. That is, if for all $A, B : \U_i$, $C : A \rightarrow B \rightarrow \U_i$ we have:
\[
  A = \sum_{a : A} 1 \qquad\text{}\qquad 
  \sum_{a:A}\sum_{b:B} C\; a\; b = \sum_{b:B}\sum_{a:A} C\; a\; b \qquad\text{}\qquad 
  isContr(A) \to A = 1
\]
then there exist terms $\mathit{unit}$ and $\mathit{flip}$, with types as in Table \ref{table:axioms}, for which the following equalities hold:
\[\coerce\; \mathit{unit}\; a = (a , *) \qquad\text{}\qquad \coerce\; \mathit{flip}\; (a,b,c) = (b , a , c) \]
for all $a:A$, $b:B$ and $c:C\, a\, b$.
\end{conjecture}


\begin{theorem}
	In the presence of function extensionality, Conjecture \ref{prop} and Conjecture \ref{conj2} are logically equivalent.
\end{theorem}
\begin{proof}
	For the forwards direction, assume function extensionality, \ref{prop} and axioms (1)-(3). By Theorem \ref{ua} we deduce that naive univalence, $\UA_i$, holds. Therefore by our assumption of \ref{prop} we deduce the proper univalence axiom for $\U_i$. Hence, by Corollary \ref{cor}, we deduce axioms (1)-(5) (possibly with different proof terms than our existing assumptions of axioms (1)-(3)). Therefore the conclusion of \ref{conj2} holds. 
	
	For the reverse direction, assume function extensionality, \ref{conj2} and naive univalence. By Theorem \ref{ua} we deduce that axioms (1)-(3) hold. Therefore by our assumption of \ref{conj2} we deduce axioms (4)-(5) also hold. Hence, by Corollary \ref{cor}, we deduce the proper univalence axiom.
\end{proof}

This result may be useful in tackling the open question of whether Conjecture \ref{prop} is a theorem of type theory, or whether there are in fact models in which it does not hold. This is because finding models where the conclusions of Conjecture \ref{conj2} do not hold given the assumptions, or showing that no such models exist, seems an easier task. For example, consider the case where the first conclusion fails, that is, where $\coerce\, \mathit{unit}\, a \neq (a,*)$ for some $A : \U$ and $a : A$. If this is the case then we have $\fst \circ (\coerce\; \mathit{unit}) : \{A : \U\} \to A \to A$ which is not equal to the identity function. We note that the existence of such a term has interesting consequences relating to parametricity and excluded middle \cite{BooijELS17}, and potentially informs our search about the type of models which might invalidate Conjecture \ref{prop}. However, we leave further investigation of this problem to future work.

\paragraph{Acknowledgements}
We would like to thank the anonymous referees for their insightful comments which materially improved the quality of the paper.

\bibliographystyle{plain}

\end{document}